\documentclass[11pt]{article}

\usepackage[margin=1in]{geometry}
\usepackage{amsmath,amssymb,amsthm,mathtools,bm}
\usepackage{booktabs}
\usepackage[numbers,sort&compress]{natbib}
\usepackage[hidelinks,breaklinks=true]{hyperref}

\newtheorem{theorem}{Theorem}
\newtheorem{proposition}{Proposition}
\newtheorem{corollary}{Corollary}
\newtheorem{lemma}{Lemma}
\theoremstyle{definition}

\newtheorem{assumption}{Assumption}
\theoremstyle{remark}
\newtheorem{remark}{Remark}

\newcommand{\calT}{\mathcal{T}}
\newcommand{\calR}{\mathcal{R}}
\newcommand{\calM}{\mathcal{M}}
\newcommand{\calP}{\mathcal{P}}
\newcommand{\calF}{\mathcal{F}}
\newcommand{\EE}{\mathbb{E}}
\newcommand{\RR}{\mathbb{R}}
\newcommand{\eps}{\varepsilon}

\hypersetup{
  pdftitle={Information Requirements for Service Allocation and Aggregate Verification},
  pdfauthor={Joss Armstrong}
}

\begin{document}

\title{Information Requirements for Service Allocation\\and Aggregate Verification}

\author{Joss Armstrong\\
  Ericsson Ireland, Athlone\\
  \texttt{joss.armstrong@ericsson.com}}

\date{September 2026}

\maketitle

\begin{abstract}
A service system may use the same categories to assign standard allocations and
to check whether each service is fulfilled. Finer categories can match
individual needs more closely, but they divide the observations available for
monitoring. We study this conflict for a fixed menu from which participants
select by declaring a category, with fulfilment assessed from each category's
aggregate outcomes during a fixed period. We relate allocation loss to variation
in preferred allocations within categories and identify conditions under which
aggregate observations preserve the verification performance of individual
records. For nested refinements under stated utility and observation
assumptions, an allocation-loss tolerance and a per-category detection target
define a feasibility band. Categories must be fine enough to provide suitable
allocations but sufficiently populated to support verification. A
source-dependent lower bound on declaration entropy and a minimum contributor
requirement give necessary information and population constraints. For an
explicit finite population with quadratic utility and binary service outcomes,
we prove the exact feasible range across all categorical designs and exhibit
designs attaining the information lower bound at specified tolerances. The
results provide conditions for choosing categories jointly for allocation and
verification.
\end{abstract}


\section{Introduction}
\label{sec:intro}
Consider a provider offering several standard service configurations. A
participant knows its own requirements and selects one of the offered
configurations. Participants selecting the same option form a category. The
provider then monitors each category separately. For example, during a fixed
service period it may record how many attempts fail that category's stated
requirement. Serving the participant requires its selected label, while assessing
fulfilment requires evidence about the service associated with that label.

The category system therefore performs two jobs. It determines which allocation
a participant receives, and it determines which observations are pooled for a
verification decision. These jobs need not favour the same level of detail. With
few categories, participants with different preferred allocations receive the
same profile. Splitting categories can improve that fit. Under the observation
conditions studied here, however, splitting also leaves smaller populations
supporting the category-level tests. A system can distinguish requirements finely
enough to allocate well but too finely to verify every service reliably within
the available period.

We ask when both requirements can be met. Allocation quality is measured by the
utility lost relative to the allocation best suited to each participant.
Verification quality is the probability of detecting a specified service
degradation at a fixed false-alarm allowance. These are different questions. A
service can match its specification while being a poor fit for an individual,
and a well-chosen service can be delivered poorly. The analysis keeps the loss
from categorical allocation separate from the statistical evidence of service
fulfilment.

The starting point is the individual's selection of means relative to its own
end. The coordinator uses that choice rather than requiring a complete account of
the individual's demand. For quantitative analysis, we specify the utility model,
population weights, and observation laws. Each candidate design specifies a menu
that remains fixed during use. We compare these designs mathematically; learning
a menu from data is not part of the analysis.

The common category system creates a joint allocation--verification constraint.
We establish a feasibility band for ordered constructions satisfying the two
comparison conditions, necessary information and population bounds for broader
sources, and an exact existence band with attained information minima for a
finite quadratic population.

Quantisation and comparison of statistical experiments provide the mathematical
ingredients. Here they are connected through a partition that selects
allocations and forms the populations on which those allocations are verified. A
short declaration-gain proposition records the supporting condition under which
participants select the intended offered profiles. The paper is organised around
the allocation loss, the verification requirement, and the resulting feasible
designs. The worked construction is followed by a numerical allocation
illustration and a discussion of related work.

\section{The mechanism}
\label{sec:mechanism}
\subsection{The menu and the lifecycle}

A menu consists of a finite set of labelled service profiles. A profile
specifies the allocation attached to a declaration. The same profile applies to
every participant selecting that label, and their membership defines the
category. The menu remains fixed for the service decision. A participant
selects among published alternatives rather than negotiating a new allocation.

The allocation model concerns fit to requirements, not a universal ranking of
service from worse to better. Under the common quadratic utility used for the
exact construction, each participant has a preferred allocation and loses
utility as the supplied profile moves away from it. For a scalar illustration, a
profile can be read as a setpoint: the largest offered value need not be the
best choice for every participant. More generally, the utility function
specifies how participants rank the offered profiles.

The coordinator first publishes the menu. A participant then declares a
category, the attached profile is assigned, and outcomes are collected during a
fixed monitoring period. The verifier judges fulfilment separately for each
category and reports its result. The selected label is sufficient to identify
the assigned profile once the menu is known. It need not identify the
participant's full demand.

The assigned profile and the observed service record have distinct roles. The
welfare calculation measures how well the fixed allocation suits the participant
when assessed under the stated utility model. The verification calculation
distinguishes the stipulated fulfilled and degraded service states from their
observations. It does not infer the participant's requirements or add an
unmodelled malfunction penalty to the allocation-loss formula.

During a declaration decision, a participant may choose any offered category.
Changing the declaration changes only the selected profile and the terms already
attached to it. Any charge or participation option used in the comparison is
part of those stated terms. Category membership does not change the offered
profiles or the payoffs through rationing or congestion. This fixed-profile
premise is what allows us to study the information and monitoring consequences of
category design without introducing a separate resource-clearing mechanism.

Menu design and profile estimation precede this decision. The mathematical
analysis and the constructed examples may use a demand distribution to assess or
construct a menu. Once the menu is established, allocating its profiles uses
declarations rather than individual demand reports.

\subsection{Notation and timing}

Let demand be $T$ with values in $\calT$, let $U(t,r)$ be the payoff from
allocation $r$ to a participant of type $t$, and let
\[
  x(t)\in\arg\max_{r} U(t,r)
\]
be the type-wise oracle allocation. The oracle is an analytical comparator. It
is not information the coordinator is assumed to hold while delivering service.
Allocations and preferred allocations take values in $\RR^d$, with the Euclidean
norm used below.

Fix a finite menu $\calM=(r_1,\dots,r_K)$ and a measurable designated category
rule $C=c(T)$, with a fixed deterministic tie convention. The category rule
specifies the allocation used in the analysis. A designated rule describes
unrestricted self-selection when its profile is preferred among those offered.
Proposition~\ref{prop:gain} states the corresponding declaration-gain
condition, and the exact construction verifies it for its actual memberships.
Write
\[
  \ell(t,r)=U(t,x(t))-U(t,r),\qquad
  D(\calM,c)=\EE\,\ell(T,r_C),\qquad
  e(\calM,c)=\EE\|x(T)-r_C\|^2 .
\]
Thus, $\ell$ is the allocation shortfall for a particular demand, $D$ is its
population expectation, and $e$ is the corresponding squared allocation error.
Welfare is expected utility under the specified model and population weights.
Those cardinal comparisons are modelling choices used for this analysis. They
do not follow from purposeful choice alone. All expectations used below are
assumed finite. When the profiles are the category means we write
\[
  r_c=\EE[x(T)\mid C=c],\qquad e(\calP)=e(\calM,c)
\]
for the partition $\calP$ induced by $c$.

We use absolute loss $D$. A normalised ratio is introduced only where its
denominator is stated. In particular, dividing by oracle welfare requires that
oracle welfare be strictly positive.

Verification uses a fixed observation period, a per-category false-alarm
allowance $a$, and a required power $p_*>a$. A false alarm declares degradation
when the service follows its fulfilled-state law. Detection power is the
probability of declaring degradation under the specified degraded-state law. The
requirement must hold for each category, so we use the power of the least
detectable category when judging a design. In a finite realised population, $M$
denotes the number of contributors and $m_c$ the actual category counts.
Welfare weights and monitoring counts coincide only in constructions that
explicitly identify them.

For deterministic declarations, category entropy is $H(C)=I(T;C)\le\log_2 K$
bits. A fixed-length label requires $\lceil\log_2 K\rceil$ bits. These
quantities describe the declaration channel only; monitoring records are
accounted for separately.

When applying the occupancy ceiling, $K$ counts the categories required to meet
the monitoring target. The common construction uses nonempty categories. An
empty category that is still required to meet the same target has no evidence
and fails $p_*>a$. A bound on occupied, monitored categories is not a bound on
arbitrarily many unused labels.

\subsection{Assumptions}

The utility assumptions control the cost of grouping different requirements
together. The observation assumptions specify when a record can reproduce
another record for the same fulfilled-versus-degraded comparison. Each is used
only in the results that name it.

\begin{assumption}[Interior curvature]
\label{as:curv}
The oracle allocation is interior, the gradient of $U(t,\cdot)$ vanishes there,
and the negative utility Hessian has eigenvalues in $[\mu,L]$ with $0<\mu\le L$
along the segments used in the comparison.
\end{assumption}

\begin{assumption}[Common quadratic utility]
\label{as:quad}
$U(t,r)=B(t)-s\|x(t)-r\|^2$ for some $s>0$ and some $B$.
\end{assumption}

Assumption~\ref{as:curv} bounds how quickly utility falls away from the
preferred allocation. Assumption~\ref{as:quad} makes that loss exactly
proportional to squared distance, with $\mu=L=2s$. It is the
matching-preferences model used by the exact construction, rather than an
assumption that every participant always prefers a larger resource allocation.

\begin{assumption}[State-independent observation]
\label{as:obs}
Let $\mathcal H\in\{0,1\}$ indicate fulfilled or degraded service, let $A$ be
the directly observed aggregate record, and let $Y$ be the comparison
individual record. There is a single randomisation rule, not depending on
$\mathcal H$, that generates $Y$ from $A$. Equivalently
$\mathcal H\to A\to Y$ is a Markov chain.
\end{assumption}

Operationally, once $A$ is known, the remaining variation in $Y$ carries no
further evidence about the service state. The aggregate observer can simulate
the comparison record without knowing whether the service was fulfilled. This
relation must be established for the chosen observations. It is not implied
merely by calling a record an aggregate.

\begin{assumption}[Pool reproduction]
\label{as:repro}
For a fixed service and $m'\le m$, there is a randomisation that transforms the
$m$-contributor aggregate record into a record with exactly the law of the
$m'$-contributor record under both service states, without knowledge of the
state.
\end{assumption}

This is a comparison between pool sizes. The larger aggregate must contain
enough information to simulate the smaller pool under either service state. The
failure-count experiment below supplies a direct example. For a refinement that
changes the service profiles, the verification-ordering lemma requires this
relation between the particular parent and child experiments being compared.

\section{Welfare}
\label{sec:welfare}
Participants placed in one category receive one profile despite having different
preferred allocations. The first result measures the resulting loss. Under
bounded utility curvature, squared allocation error controls welfare loss from
above and below. Under common quadratic utility, the two are exactly
proportional.

\begin{theorem}[Welfare Bound]
\label{thm:welfare}
Under Assumption~\ref{as:curv}, for any menu and any category rule,
\begin{equation}
\label{eq:W}
  \frac{\mu}{2}\,e(\calM,c)\;\le\;D(\calM,c)\;\le\;\frac{L}{2}\,e(\calM,c).
\end{equation}
\end{theorem}

\begin{proof}
Fix $t$ and expand $U(t,\cdot)$ from $x(t)$ to the delivered profile $r_{c(t)}$.
The gradient vanishes at $x(t)$, so the linear term is zero. The quadratic
remainder is
$\tfrac12 (r_{c(t)}-x(t))^\top \nabla^2_{rr}(-U)(t,\tilde r)\,(r_{c(t)}-x(t))$
for some $\tilde r$ on the segment, and its eigenvalues lie in $[\mu,L]$. Hence
$\tfrac{\mu}{2}\|x(t)-r_{c(t)}\|^2 \le \ell(t,r_{c(t)}) \le
\tfrac{L}{2}\|x(t)-r_{c(t)}\|^2$. Take expectations over $T$.
\end{proof}

Two consequences are worth stating separately, because they concern the two
design rules that appear later.

\begin{corollary}[Mean profiles]
\label{cor:mean}
Let $\calP$ be a partition, let $D_{\mathrm{mean}}(\calP)$ be the loss of the
mean-profile rule on it, and let $D_{\mathrm{opt}}(\calP)$ be the smallest loss
among all rules that are constant on each cell. Under
Assumption~\ref{as:curv},
\begin{equation}
\label{eq:Wmean}
  \frac{\mu}{2}\,e(\calP)\;\le\;D_{\mathrm{opt}}(\calP)\;\le\;
  D_{\mathrm{mean}}(\calP)\;\le\;\frac{L}{2}\,e(\calP).
\end{equation}
Under Assumption~\ref{as:quad} the mean rule is optimal on its partition and
$D=s\,e$ exactly.
\end{corollary}

\begin{proof}
The conditional mean minimises squared error among cell-constant profiles, so
every cell-constant candidate has squared error at least $e(\calP)$ and, by the
lower bound in~\eqref{eq:W}, loss at least $\tfrac{\mu}{2}e(\calP)$. The mean
rule is one such candidate, which gives the middle inequality. The upper bound
is~\eqref{eq:W} applied to the mean rule. Under Assumption~\ref{as:quad} the
loss is $s\|x(t)-r\|^2$ pointwise, so minimising loss on a cell is minimising
squared error on it, and the two coincide with $\mu=L=2s$.
\end{proof}

\begin{corollary}[Refinement]
\label{cor:refine}
Let $\calP'$ refine $\calP$, both with mean profiles. Then
\begin{equation}
\label{eq:Wrefine}
  e(\calP)-e(\calP')=\EE\|r_{C'}-r_C\|^2\;\ge\;0 .
\end{equation}
\end{corollary}

\begin{proof}
Each finer cell sits inside one coarser cell, so
$x-r_C=(x-r_{C'})+(r_{C'}-r_C)$. Square and take expectations. Conditioning the
cross term on $C'$ gives zero, because $r_{C'}$ is the conditional mean of $x$
given $C'$ and $r_C$ is $C'$-measurable.
\end{proof}

Under common quadratic utility, refinement with mean profiles cannot increase
the actual allocation loss. For general utility, the best cell-constant policy
also cannot become worse, because it can retain the coarse policy after a split.
The arithmetic-mean implementation need not coincide with that best policy. Its
curvature bounds remain valid, but they alone do not establish monotonicity of
its actual non-quadratic loss.

\subsection{Declaration gain}
\label{sec:gain}

A fixed menu also permits a simple check on the designated category rule. Define
the gain from an alternative declaration by
\[
  g(t)=\max_{1\le k\le K}U(t,r_k)-U(t,r_{c(t)}),
\]
where the current declaration is included among the alternatives.

\begin{proposition}[Declaration Gain]
\label{prop:gain}
For a fixed offered menu and any designated category rule,
\[
  0\le g(t)\le\ell(t,r_{c(t)}),
  \qquad
  \EE g(T)\le D(\calM,c).
\]
For mean profiles under the Welfare Bound's assumptions, $\EE g(T)\le
L\,e(\calP)/2$. The gain is zero exactly when the designated profile maximises
the participant's utility over the offered menu. Under common quadratic
utility, nearest-profile selection has this property.
\end{proposition}

\begin{proof}
No offered profile provides more utility than the participant's oracle
allocation. Subtracting the utility of the designated profile proves the
pointwise upper bound, and averaging proves the expected bound. The current
declaration gives the lower bound of zero. A utility-maximising declaration
admits no profitable alternative. Under common quadratic utility, maximising
utility is exactly minimising distance to the preferred allocation. A
randomised declaration cannot outperform the best offered profile.
\end{proof}

The proposition bounds the gain from choosing another existing profile. It does
not assert that the menu itself is optimal.
Appendix~\ref{app:fixedmenu} records two fixed-menu consequences of it, an
equality between the constrained and unconstrained optima at a given menu size
and an exact expression for the expected gain. Neither is used later.

\section{Aggregate verification}
\label{sec:aggdom}
Allocation asks whether a profile fits a participant's requirements.
Verification asks whether the observed service is consistent with fulfilment of
that profile. Fix the service, the monitoring period, and the fulfilled and
degraded observation laws.

Let $A$ be the aggregate record and $Y$ the comparison individual record. Under
the state-independent observation assumption, an observer with $A$ can generate
a record distributed like $Y$ under either service state. The additive model
$Y_{ij}=A_i+\xi_{ij}$ is one instance when the added noise is independent of the
aggregate and the service state. The aggregate itself may be noisy.

Write $\pi^*_A(a)$, $\pi^*_Y(a)$, and $\pi^*_{(A,Y)}(a)$ for the best attainable
powers at false-alarm allowance $a$, allowing randomised tests and taking
suprema where needed.

\begin{theorem}[Aggregate Dominance]
\label{thm:aggdom}
Under Assumption~\ref{as:obs}, for every $a\in[0,1]$,
\begin{equation}
\label{eq:AD}
  \pi^*_{(A,Y)}(a)=\pi^*_A(a)\;\ge\;\pi^*_Y(a).
\end{equation}
\end{theorem}

\begin{proof}
Take any test using the individual record. An aggregate observer can generate a
simulated individual record with the stated randomisation rule and run that
test. The simulated record has the same distribution as the actual individual
record under fulfilled service and under degraded service. False-alarm
probability and power are therefore unchanged. Every individual-record test can
be matched from the aggregate, giving the inequality after optimisation. The
same simulation reproduces any test using both records. Since an observer with
both records may also ignore the individual record, their optimal power equals
that available from the aggregate alone.
\end{proof}

This is the classical comparison-of-experiments argument
\citep{blackwell1953}. Its role here is to establish when category-level
fulfilment can be checked without additional individual detail. The theorem
gives weak dominance. The Gaussian specialisation in Appendix~\ref{app:aggdom}
identifies strictness conditions and quantifies the gap. The equality for an
observer holding both records also makes clear that merely possessing extra
information cannot reduce optimal performance.

\section{Information budget feasibility}
\label{sec:budget}
The same categories determine the allocation rule and the populations supplying
evidence about service fulfilment. We now hold the population, observation
period, degradation specification, and error criterion fixed while comparing
category designs. An allocation may become more accurate as categories are
refined, yet verification must still work separately for every category. The
relevant monitoring performance is therefore the power of the least detectable
category, not an average across categories.

\subsection{The band}

Fix a finite ordered sequence of constructions on the same population, indexed
by $j$. Write $D_j$ for actual welfare loss at level $j$, and $\pi_j=\min_c
\pi_{j,c}$ for the least optimal power among the occupied categories, at a
per-category false-alarm allowance that is not enlarged as $j$ grows.

The opposing effects follow from two reproduction arguments. A refined
allocation rule can retain the parent's profile before improving the fit within
each child. An observer with a parent record can reproduce a child record when
the observation assumption holds. We state these comparisons before taking their
intersection.

\begin{lemma}[Welfare ordering]
\label{lem:worder}
Let the categories form a nested refinement on a fixed population, with mean
profiles. Under Assumption~\ref{as:quad}, $D_{j+1}\le D_j$.
\end{lemma}

\begin{proof}
Corollary~\ref{cor:refine} gives $e(\calP_{j+1})\le e(\calP_j)$, and
Corollary~\ref{cor:mean} gives $D=s\,e$ exactly under
Assumption~\ref{as:quad}.
\end{proof}

\begin{lemma}[Verification ordering]
\label{lem:vorder}
Let the categories form a nested refinement, and suppose each child's stipulated
fulfilled/degraded aggregate experiment satisfies Assumption~\ref{as:repro}
relative to its parent's, at a false-alarm allowance no larger. Then
$\pi_{j+1}\le\pi_j$.
\end{lemma}

\begin{proof}
Let $c$ be a category attaining $\pi_j$. It has at least one nonempty child $c'$.
An observer holding the aggregate of $c$ can simulate the aggregate of $c'$, by
Assumption~\ref{as:repro}, and then run any $c'$-based test. Size and power are
preserved, so $\pi_{j+1,c'}\le\pi_{j,c}$. Shrinking the allowance only removes
tests. The minimum at level $j+1$ ranges over a set containing $c'$, so
$\pi_{j+1}\le\pi_{j+1,c'}\le\pi_{j,c}=\pi_j$.
\end{proof}

Thus the stated nested construction has non-increasing allocation loss and
non-increasing guaranteed verification power. The following interval result also
applies to any other ordered family for which these same two orderings have been
established.

\begin{theorem}[Information Budget Feasibility]
\label{thm:budget}
Suppose $D_{j+1}\le D_j$ and $\pi_{j+1}\le\pi_j$ along the sequence, and fix a
welfare tolerance $\eps$ and a power target $p_*$. If $\{j:D_j\le\eps\}$ or
$\{j:\pi_j\ge p_*\}$ is empty, no level meets both requirements. Otherwise, with
$j_W=\min\{j:D_j\le\eps\}$ and $j_V=\max\{j:\pi_j\ge p_*\}$, the levels meeting
both are exactly
\begin{equation}
\label{eq:BAND}
  \calF=\{j:\;j_W\le j\le j_V\},
\end{equation}
which is itself empty when $j_W>j_V$.
\end{theorem}

\begin{proof}
Welfare acceptance is an upper set, since $D$ is non-increasing, so it is exactly
$\{j\ge j_W\}$. Verification acceptance is a lower set, since $\pi$ is
non-increasing, so it is exactly $\{j\le j_V\}$. Their intersection is the stated
interval.
\end{proof}

Find the first level that meets the allocation-loss tolerance. If that level
also meets the verification target, it is the coarsest feasible level. If it
fails verification, every finer level fails too and the band is empty. The
theorem describes the stated ordered family. A converse across all alternative
designs requires a separate attaining construction, supplied below.

The verification ordering permits unequal category sizes and flat stretches of
power. What matters is the relation between the specified parent and child
experiments. No Gaussian calculation is needed for this comparison.

\subsection{An occupancy ceiling}

A failure-count example makes the pool-size comparison concrete. During a fixed
period, each contributor supplies one independent indicator of failure to meet
the category's service requirement. The fulfilled and degraded states have
failure probabilities $p_0$ and $p_1>p_0$. The verifier sees the contributor
count and the total number of failures, rather than identified individual
records.

Given $S_m=s$ failures among $m$ contributors, form $m$ markers, $s$ of them
failures, draw $m'$ of them uniformly without replacement, and count failures
among the draw. Every binary record with exactly $s$ failures has probability
$p^{s}(1-p)^{m-s}$, so conditional on $s$ all arrangements are equally likely
whatever $p$ is. The larger aggregate can therefore reproduce the smaller
aggregate under either state. It also reproduces the full indicator record
conditional on its count by distributing the failure markers uniformly over the
recorded positions. This supplies the observation premise of Aggregate Dominance
for this count experiment as well as the pool-reproduction premise.
Appendix~\ref{app:budget} gives the conditional law and the exact optimal
count-test powers.

Now suppose categories share this monitoring experiment, or another specified
experiment with the same pool-reproduction property. At the fixed false-alarm
allowance, let $q$ be the smallest number of contributors that achieves the
required power, assuming such a finite number exists. The service period and
degradation are held fixed when $q$ is defined. Write $\Pi(m,a)$ for the optimal
power with $m$ contributors. Assumption~\ref{as:repro} makes $\Pi$ non-decreasing
in $m$, by the argument of Lemma~\ref{lem:vorder} applied to pool sizes. With
$\Pi(0,a)=a$ and $p_*>a$, define the finite requirement
\[
  q=\min\{m\ge1:\Pi(m,a)\ge p_*\}.
\]

\begin{corollary}[Occupancy ceiling]
\label{cor:occ}
Every category meets the power target exactly when its actual count is at least
$q$. Hence
\begin{equation}
\label{eq:OCC}
  \min_c m_c\ge q
  \quad\Longrightarrow\quad
  Kq\le M,
  \qquad K\le\lfloor M/q\rfloor .
\end{equation}
\end{corollary}

\begin{proof}
Monotonicity of $\Pi$ gives the first claim. Summing the $K$ per-category
requirements over a population of size $M$ gives $Kq\le M$.
\end{proof}

The cardinality cap limits category count but does not certify an arbitrary
partition below the limit. The actual check is that its smallest category also
has at least $q$ members. Service-specific contributor requirements and
family-wise false-alarm control are treated separately in
Appendix~\ref{app:budget}.

At each resolution, the alternative is degradation of the affected category's
own service. It is not a fixed physical subset of incidents carried unchanged
through different partitions.

\subsection{An information envelope}

Let $X=x(T)$. For the finite weighted sources and bounded-density continuous
sources covered in Appendix~\ref{app:floor}, the welfare assumptions give an
explicit source-dependent floor $R_P(\eps)$: every category/profile design with
loss at most $\eps$ must have $H(C)\ge R_P(\eps)$. This floor concerns the source
and tolerance, not the output of one clustering algorithm.
Appendix~\ref{app:floor} states it for finite sources with weights and for
bounded-density continuous sources. Together with~\eqref{eq:OCC} this bounds
the declaration entropy on both sides.

\begin{corollary}[Information envelope]
\label{cor:info}
Every jointly acceptable design satisfies
\begin{equation}
\label{eq:INFO}
  R_P(\eps)\;\le\;H(C)\;\le\;\log_2\lfloor M/q\rfloor .
\end{equation}
\end{corollary}

\begin{proof}
The source floor gives the left inequality. The right follows from
$H(C)\le\log_2 K$ and Corollary~\ref{cor:occ}.
\end{proof}

If the lower information requirement exceeds the population-based upper bound,
no design in the stated class meets both obligations. Overlap provides a
necessary condition, not an attaining design. The next section constructs an
exact feasible range for one finite source and identifies tolerances at which
the information lower bound is attained.

\section{An exact allocation--verification construction}
\label{sec:construction}
We now exhibit one population for which the allocation limit and verification
ceiling are attained by the same category designs. The comparison covers every
partition and every profile choice for that population, not just a selected
refinement hierarchy. The resulting band therefore distinguishes attainable from
impossible category counts. The profiles also support the intended memberships
under participants' own choices, as established in Appendix~\ref{app:construction}.

\subsection{The population and the design}

Let $M$ equally weighted agents have demands $t_i=(i+\tfrac12)/M$ for
$i=0,\dots,M-1$, and let $U(t,r)=1-(t-r)^2$, so $x(t)=t$ and
Assumption~\ref{as:quad} holds with $s=1$. Write
$D=\tfrac1M\sum_i (t_i-r_{c(i)})^2$.

A design partitions the agents into $K$ nonempty categories of sizes
$m_1,\dots,m_K$ and delivers one profile per category. The comparison class
allows \emph{every} partition and \emph{every} choice of profiles, so lower
bounds proved on it apply as a special case to designs whose memberships arise
from participants' own choices.

Putting several distinct preferred allocations into one category forces them to
share one profile. The least costly group of a given size consists of
consecutive demands served at their mean. The next calculation makes that cost
explicit.

\begin{lemma}[Grouping cost]
\label{lem:group}
A category holding $m$ of these demands incurs squared-error sum at least
$m(m^2-1)/(12M^2)$, whatever profile it is given, with equality for $m$
consecutive demands and their mean.
\end{lemma}

\begin{proof}
For points $x_1,\dots,x_m$ with mean $\bar x$ and any $r$,
$\sum_i(x_i-r)^2=\sum_i(x_i-\bar x)^2+m(\bar x-r)^2\ge
\tfrac1m\sum_{i<j}(x_i-x_j)^2$. Sort the selected grid points. The gap between
the $j$th and $i$th is at least $(j-i)/M$, so the last sum is at least
$\tfrac1{mM^2}\sum_{i<j}(j-i)^2=m(m^2-1)/(12M^2)$. Consecutive points with their
mean attain both steps.
\end{proof}

Summing Lemma~\ref{lem:group} over categories gives, for every design,
\begin{equation}
\label{eq:grouplb}
  D\;\ge\;\frac{\sum_c m_c^3-M}{12M^3}.
\end{equation}

\begin{lemma}[Balanced sizes]
\label{lem:balance}
For integers $b\ge a+2$, $a^3+b^3-\bigl[(a+1)^3+(b-1)^3\bigr]=3(a+b)(b-a-1)>0$.
Hence at fixed $K$ and $M$, $\sum_c m_c^3$ is minimised when the sizes differ by
at most one.
\end{lemma}

Take $K$ consecutive groups with sizes differing by at most one, and deliver each
group's mean. This attains~\eqref{eq:grouplb}, so it is a global minimum over all
partitions and all profiles, not only over some fixed hierarchy. Write $D_K^*$
for that minimum. When $K\mid M$,
\begin{equation}
\label{eq:Dstar}
  D_K^*=\frac1{12K^2}-\frac1{12M^2},
\end{equation}
and the adjacent-size formula for the remaining $K$ is in
Appendix~\ref{app:construction}.

These memberships can be implemented by category choice. The boundary between
adjacent mean profiles lies between the corresponding neighbouring demand points,
so every participant prefers its own group's profile. Appendix~\ref{app:construction}
gives the boundary calculation. No redistribution of participants after their
choices is needed.

\subsection{The exact band}

Fix a welfare tolerance $\eps$ and the monitoring requirement $q\le M$ of
Corollary~\ref{cor:occ}.

\begin{theorem}[Exact existence band]
\label{thm:exactband}
A deterministic $K$-category design meeting the welfare tolerance and the
every-category power target exists if and only if
\begin{equation}
\label{eq:EXACTBAND}
  K_W(\eps)\;\le\;K\;\le\;\lfloor M/q\rfloor,
  \qquad K_W(\eps)=\min\{K:D_K^*\le\eps\}.
\end{equation}
At every $K$ in that range the nearly balanced consecutive centroid design is
globally welfare optimal, meets the power target, and places every agent nearest
its own group's profile, so its designed occupancies are the occupancies that
choice produces.
\end{theorem}

\begin{proof}
Necessity. Any $K$-category design has loss at least $D_K^*$, so $K<K_W$ fails
welfare. Every monitored category needs $q$ agents, so $K>\lfloor M/q\rfloor$
fails monitoring by Corollary~\ref{cor:occ}. Achievability. For $K$ in the range
the balanced design has loss $D_K^*\le\eps$, its smallest group has
$\lfloor M/K\rfloor\ge q$ members, and Lemma~\ref{lem:selfsel} gives zero
declaration gain, so those occupancies are the ones choice realises.
\end{proof}

The quantifiers matter. The statement is that some design attains the targets at
every $K$ in the band, and that no design attains them outside it. It does not
say every design inside the band works.

\subsection{Two consequences}

\begin{corollary}[Verification enforces a positive welfare floor]
\label{cor:floor}
Every verifiable design has $D\ge (q^2-1)/(12M^2)$, with equality when $q\mid M$
using consecutive $q$-member groups. The exact global minimum verifiable loss is
$D^*_{\lfloor M/q\rfloor}$. For $q>1$ it is strictly positive, so joint
perfection is impossible.
\end{corollary}

\begin{proof}
Lemma~\ref{lem:group} makes the average loss inside a size-$m$ category at least
$(m^2-1)/(12M^2)$, and every verifiable category has $m\ge q$. Averaging over
membership gives the bound. The exact minimum follows from the occupancy ceiling,
the strict decrease of $D_K^*$ below $K=M$, and Theorem~\ref{thm:exactband}.
\end{proof}

\begin{corollary}[Entropy floor]
\label{cor:entropy}
Let $p_c=m_c/M$. Every design on this population satisfies
\begin{equation}
\label{eq:entlb}
  D\;\ge\;\frac1{12}\Bigl(2^{-2H(C)}-M^{-2}\Bigr).
\end{equation}
\end{corollary}

\begin{proof}
Rewrite~\eqref{eq:grouplb} as $12D\ge\sum_c p_c^3-M^{-2}$. The weighted
arithmetic--geometric mean inequality applied to the numbers $p_c^2$ with weights
$p_c$ gives $\sum_c p_c^3=\sum_c p_c\,p_c^2\ge\prod_c (p_c^2)^{p_c}=2^{-2H(C)}$.
\end{proof}

This is a genuine converse over all partitions, with unequal category
probabilities permitted. It is not obtained by substituting a lower bound on $K$
for a lower bound on $H(C)$. At a balanced operating point it is tight. If
$K\mid M$ and $\eps=D_K^*$, then~\eqref{eq:entlb} forces $H(C)\ge\log_2 K$, and
the balanced $K$-category design attains exactly $\log_2 K$ with zero declaration
gain. When $M/K\ge q$ the same design also meets verification. So $\log_2 K$ is
the minimum category entropy across all categorical designs meeting both targets,
not merely the encoding cost of a rule chosen in advance.

\subsection{The worked numbers}
\label{sec:worked}

Take $M=96$, one independent binary failure observation per agent over a fixed
period, $p_0=1/5$ under fulfilled service and $p_1=4/5$ under a category-wide
degradation, per-category allowance $a=1/10$ and target $p_*=4/5$. Exact
likelihood-ratio ordering with randomisation at the boundary count gives optimal
powers $2/5$, $7/10$ and $22/25$ at $m=1,2,3$, so $q=3$ and
$\lfloor M/q\rfloor=32$.

At $\eps=0.002$ we have $D_6^*=85/36864\approx 0.0023058$ and
$D_7^*=1001/589824\approx 0.0016971$, so $K_W=7$ and the band
of~\eqref{eq:EXACTBAND} is exactly $\{7,\dots,32\}$. Table~\ref{tab:band}
displays the endpoints and three interior points.

\begin{table}[t]
\centering
\caption{The exact band on the 96-agent construction at $\eps=0.002$, $q=3$.
Loss is the global minimum $D_K^*$ over all partitions and profiles. Power is
the optimal count-test power at the design's smallest category.}
\label{tab:band}
\begin{tabular}{rrrrl}
\toprule
$K$ & smallest pool & $D_K^*$ & power & status \\
\midrule
 6 & 16 & 0.00230577 & 0.99997186 & welfare fails \\
 7 & 13 & 0.00169712 & 0.99983484 & feasible \\
 8 & 12 & 0.00129304 & 0.99952596 & feasible \\
16 &  6 & 0.00031648 & 0.98311000 & feasible \\
32 &  3 & 0.00007234 & 0.88000000 & feasible \\
33 &  2 & 0.00006951 & 0.70000000 & verification fails \\
\bottomrule
\end{tabular}
\end{table}

The minimum verifiable loss is $D_{32}^*=1/13824\approx 7.2338\times10^{-5}$, the
positive floor of Corollary~\ref{cor:floor}. At $\eps=0.002$ the envelope
of~\eqref{eq:INFO} reads $2.68716\le H(C)\le 5$, and the seven-category design
uses $2.80656$ bits with five groups of $14$ and two of $13$.

The information lower bound is attained at a tighter allocation tolerance. At
$\eps=D_8^*=143/110592$, the quantity $12\eps+96^{-2}$ equals $1/64$, so the
entropy-floor inequality requires at least three bits. The balanced
eight-category design uses exactly three bits and meets the verification target.
Three bits of category entropy are therefore necessary and sufficient at this
specified pair of targets. The upper entropy limit remains five bits, so
attainment of the lower bound does not collapse the whole information envelope to
a single point. Allocation losses and count-test powers were evaluated in exact
rational arithmetic. Entropy values were evaluated from the exact category
probabilities using base-two logarithms. The three-bit equality follows
algebraically.

\section{Numerical illustration}
\label{sec:numerical}
The exact construction uses a uniform scalar population. We also illustrate
allocation loss for a synthetic population with four-dimensional requirements.
Each of ten runs, using seeds 1--10, draws 50,000 demand vectors from a mixture
of five service types, with coordinates normalised to $[0,1]$. The utility is
$U(t,r)=-\|t-r\|^2$, so allocation loss equals squared error directly for every
delivered profile.

At $K=3$, the service types are grouped by resource intensity and each group
receives its empirical mean. At $K=10$ and $K=30$, the categories are fitted by
$k$-means and participants receive the profiles returned by the fitted model.
Table~\ref{tab:numerical} reports absolute allocation loss and that loss divided
by total demand variance. The latter ratio is a scale-normalised squared error,
not a ratio to oracle welfare. Oracle welfare is zero for the unshifted utility
used here.

\begin{table}[t]
\centering
\caption{Ten runs of $50{,}000$ synthetic demand vectors. $D$ is mean squared
error against the delivered profiles and equals absolute allocation loss for the
stated utility. $D/\mathrm{Var}(T)$ is shown with one standard deviation across
runs. The gain columns evaluate alternatives within the same fixed menu and serve
as a check on Proposition~\ref{prop:gain}. Across the ten runs, $D$ and mean~$g$
are reported as the mean of the per-run values; max~$g$ is the overall maximum
across all runs.}
\label{tab:numerical}
\begin{tabular}{clrlrr}
\toprule
$K$ & category rule & $D$ & $D/\mathrm{Var}(T)$ & mean $g$ & max $g$ \\
\midrule
 3 & semantic       & 0.18625 & $0.5616\pm0.0019$ & $1.93\times10^{-2}$ & $8.99\times10^{-1}$ \\
10 & demand-derived & 0.04200 & $0.1266\pm0.0007$ & $0$ & $0$ \\
30 & demand-derived & 0.01647 & $0.0497\pm0.0005$ & $0$ & $0$ \\
\bottomrule
\end{tabular}
\end{table}

For run $b$, the demand variance is
\[
  V_b=\frac1N\sum_{i=1}^N\|t_{bi}-\bar t_b\|^2,
\]
and the reported normalised quantity is computed per run as $D_b/V_b$. The mean
and standard deviation of this ratio across the ten runs are reported.

Loss is lower for the two fitted demand-derived designs, and lower at 30 than at
10 categories. These are observed results for the stated fitting procedure. The
fitted partitions are not established as a nested sequence, so their ordering is
not an application of the refinement corollary. Nor does the comparison with
three semantic groups isolate category construction at a matched value of $K$.

The gain columns provide a supporting check on Proposition~\ref{prop:gain}. Gain
is computed against the same fixed profiles used for delivery. It is zero within
the stated tolerance for all evaluated participants at $K=10$ and $K=30$. The
semantic comparator has positive mean gain $1.93\times10^{-2}$, below its
allocation loss as the proposition requires. Appendix~\ref{app:repro} gives the
distance-based evaluation and numerical convention.

Returned profiles need not equal the empirical means of their final memberships.
The profile-to-mean error decomposition in Appendix~\ref{app:repro} accounts
for this distinction. Its weighted offset is $1.1\times10^{-7}$ at $K=10$ and
$1.7\times10^{-7}$ at $K=30$, and zero for the mean-profile semantic comparator.
All reported allocation errors use the profiles actually delivered.

No fulfilment experiment is run on this synthetic population. The exact
construction supplies the worked verification band. Table~\ref{tab:numerical}
illustrates the allocation quantities and the supporting declaration condition.

\section{Related work}
\label{sec:related}
\paragraph{Allocation and quantisation.}
Assigning one representative profile to a category is a quantisation problem.
Nearest-neighbour assignment and centroid representatives are associated with
Lloyd's construction \citep{lloyd1982}, and \citet{gray_neuhoff1998} survey
quantisation and distortion. \citet{shannon1948} and \citet{berger1971} provide
the rate-distortion background. Under common quadratic utility, squared-error
optimisation also optimises the allocation loss used here. Under general smooth
utility that equivalence is not guaranteed, which is why the Welfare Bound
retains separate curvature bounds. Appendix~\ref{app:floor} derives the
source-dependent entropy bounds used in the information envelope.

\paragraph{Aggregate verification.}
Blackwell's comparison of experiments \citep{blackwell1953} supplies the
state-independent simulation argument behind Aggregate Dominance.
\citet{neyman_pearson1933} give the optimal-testing principle for the
binary construction, with the randomised-test formulation as in
\citet{lehmann_romano2005}. The pool-size comparison applies the same
simulation principle to larger and smaller aggregate records. This is the
step that turns a category-level detection requirement into a population
requirement.

\paragraph{Declarations and finite messages.}
Menu choice by an informed participant is studied in delegation
\citep{khodabakhsh2024}, while bounded-communication mechanisms explicitly
constrain the message space \citep{bns2007}. Proposition~\ref{prop:gain} supplies
only the declaration-gain condition that the present fixed-profile mechanism
needs.

\paragraph{The shared category system.}
Task-oriented compression designs representations for a downstream objective
\citep{shlezinger2019,tishby1999}. The setting here assigns the category system
two linked roles. A declaration selects an allocation, and category membership
supplies the population whose outcomes support verification. The paper
establishes the consequences of that coupling under stated premises, including a
construction-relative band, general necessary bounds, and a finite setting with
an exact existence band and attained information minima.

\section{Conclusion}
\label{sec:conclusion}
Categories used for allocation also determine the evidence available for
category-level verification. Under the stated refinement and observation
conditions, finer categories reduce allocation loss while weakening the
guaranteed power of the category tests. The acceptable resolutions therefore
form a feasibility band. Its coarsest allocation-acceptable point provides a
direct feasibility check: if that point cannot be verified at the required power,
no finer point in the same construction can satisfy both requirements.

A source-dependent information floor and a minimum contributor requirement
express the two constraints separately. They provide an impossibility
certificate when their requirements do not overlap. The uniform finite
construction goes further by establishing the exact feasible range across all
categorical designs for its population and attaining the information lower bound
at specified tolerances. It also shows how a monitoring requirement can impose a
strictly positive allocation loss.

These conclusions concern fixed profiles, a stated monitoring period, and
explicitly comparable observation laws. Within that setting, choosing categories
only for allocation fit can produce services that cannot be verified with the
available evidence. Category design must therefore account for both the
distinctions needed to allocate and the populations needed to verify.

\appendix

\section{Aggregate observation}
\label{app:aggdom}
Theorem~\ref{thm:aggdom} applies whenever the stated simulation relation holds.
This appendix evaluates the comparison in the Gaussian additive model, giving
exact powers, strictness conditions, and a finite-sample bound on the power gap.
These quantitative results specialise the general theorem rather than supply a
missing part of its proof.

\subsection{Scalar case}

Assume independent periods $i=1,\dots,n$, and for fixed $\sigma>0$ and
$\delta>0$,
\[
  A_i\mid H_0\sim\mathcal N(\nu,\sigma^2),
  \qquad
  A_i\mid H_1\sim\mathcal N(\nu-\delta,\sigma^2).
\]
Let $\xi_{ij}$ be independent $\mathcal N(0,\tau^2)$, independent of $A$ and of
the state, with $m$ agents per period and $\tau\ge0$. The two grand means have
variances $\sigma^2/n$ and $(\sigma^2+\tau^2/m)/n$ and the same mean shift
$\delta$. In this model the within-period residuals $Y_{ij}-\bar Y_i$ are
independent of the period means with a state-independent law, so the individual
grand mean attains the best individual-only test and the comparison is not
weakened artificially.

With $z_{1-a}=\Phi^{-1}(1-a)$ for $0<a<1$, the noncentralities and exact optimal
powers are
\[
  \lambda_{\mathrm{agg}}=\frac{\delta\sqrt n}{\sigma},
  \qquad
  \lambda_{\mathrm{flow}}=\frac{\delta\sqrt n}{\sqrt{\sigma^2+\tau^2/m}},
  \qquad
  \pi_\bullet=\Phi(\lambda_\bullet-z_{1-a}).
\]
So $\pi_{\mathrm{agg}}\ge\pi_{\mathrm{flow}}$, strictly when $\delta>0$,
$\tau^2>0$, $m$ finite and $0<a<1$. Equality holds at $\tau^2=0$, and the gap
tends to zero as $m$ grows with the other parameters fixed. At $\delta=0$ both
powers equal $a$.

\subsection{\texorpdfstring{A finite-sample $O(1/m)$ bound}{A finite-sample O(1/m) bound}}

For $u\ge0$,
\[
  \frac1\sigma-\frac1{\sqrt{\sigma^2+u}}
  =\frac{u}{\sigma\sqrt{\sigma^2+u}\,\bigl(\sigma+\sqrt{\sigma^2+u}\bigr)}
  \;\le\;\frac{u}{2\sigma^3}.
\]
Put $u=\tau^2/m$. The normal density is at most $1/\sqrt{2\pi}$, so
\begin{equation}
\label{eq:gap1m}
  0\;\le\;\pi_{\mathrm{agg}}-\pi_{\mathrm{flow}}
  \;\le\;\frac{\delta\sqrt n\,\tau^2}{2\sqrt{2\pi}\,\sigma^3\,m}.
\end{equation}
This is a finite-$m$ inequality, not an asymptotic expansion. It gives the
$O(1/m)$ rate at fixed $n$, $\delta$, $\sigma$ and $\tau$. It does not assert a
uniform rate when those parameters vary with $m$ or with category refinement.

\subsection{Vector case}

As in the scalar case, assume independent Gaussian aggregate observations with a
common positive-definite covariance $\Sigma$ under the two service states and
mean difference $\delta$. Let the independent Gaussian individual noise have
covariance $\Sigma_\xi\succeq0$. Then
\[
  \lambda_{\mathrm{agg}}^2=n\,\delta^\top\Sigma^{-1}\delta,
  \qquad
  \lambda_{\mathrm{flow}}^2
  =n\,\delta^\top\bigl(\Sigma+\Sigma_\xi/m\bigr)^{-1}\delta,
\]
with exact powers $\Phi(\lambda_\bullet-z_{1-a})$ as before. Setting
$v=\Sigma^{-1/2}\delta$ and $G=\Sigma^{-1/2}\Sigma_\xi\Sigma^{-1/2}$,
\[
  \lambda_{\mathrm{agg}}^2-\lambda_{\mathrm{flow}}^2
  =\frac nm\,v^\top G\Bigl(I+\frac Gm\Bigr)^{-1}v\;\ge\;0,
  \qquad
  \lambda_{\mathrm{agg}}^2-\lambda_{\mathrm{flow}}^2\le\frac nm\,v^\top Gv,
\]
and for $\delta\neq0$,
\[
  0\;\le\;\pi_{\mathrm{agg}}-\pi_{\mathrm{flow}}
  \;\le\;\frac{n\,v^\top Gv}{\sqrt{2\pi}\,\lambda_{\mathrm{agg}}\,m}.
\]
Equality at finite $m$ holds exactly when $Gv=0$, that is
$\Sigma_\xi\Sigma^{-1}\delta=0$. So strictness depends on whether the individual
noise touches the detection direction. Noise entirely orthogonal to it costs
nothing.

The comparison in this appendix holds $n$ and $m$ fixed. Any substitution that
makes $m$ depend on the category count belongs to
Section~\ref{sec:budget}, where it is derived from the occupancy of the actual
construction, and it is not inferred from Theorem~\ref{thm:aggdom}.

\section{Feasibility refinements}
\label{app:budget}
\subsection{The count-thinning law}

Section~\ref{sec:budget} uses a binary aggregate experiment to show that
Assumption~\ref{as:repro} has non-Gaussian realisations. Here is the conditional
law it relies on. For $0\le s\le m$ and
$0\le m'\le m$, the number $J$ of failures in a uniformly selected
$m'$-subset of $m$ markers satisfies
\begin{equation}
\label{eq:hyper}
  \Pr(J=j\mid S_m=s)
  =\frac{\binom{s}{j}\binom{m-s}{m'-j}}{\binom{m}{m'}},
\end{equation}
with impossible combinations assigned zero. The rule uses no $p$. If
$S_m\sim\mathrm{Binomial}(m,p)$, then $J\sim\mathrm{Binomial}(m',p)$. So the
larger aggregate count generates the smaller aggregate count under both service
states, without observing identities and without knowing which state holds.
Independent repeated periods are handled by applying the rule to each period
separately.

The verifier never holds the individual records in this construction. It holds
$m$ and $S_m$, and~\eqref{eq:hyper} operates on those two numbers alone.

\subsection{Exact optimal count-test powers}

Fix $p_0<p_1$ and a per-category allowance $a$. The likelihood ratio of the $p_1$
law to the $p_0$ law is increasing in the count, so an optimal test rejects at
the largest counts first and randomises at the boundary count to spend the
allowance exactly. This is the Neyman--Pearson construction
\citep{neyman_pearson1933} for a finite sample space.

For the worked construction of Section~\ref{sec:worked}, with $p_0=1/5$,
$p_1=4/5$ and $a=1/10$, the arithmetic is exact. At $m=2$, reject at count $2$
and with probability $3/16$ at count $1$:
\[
  \text{size}=\frac1{25}+\frac3{16}\cdot\frac8{25}=\frac1{10},
  \qquad
  \text{power}=\frac{16}{25}+\frac3{16}\cdot\frac8{25}=\frac7{10}.
\]
At $m=3$, reject at count $3$ and with probability $23/24$ at count $2$:
\[
  \text{size}=\frac1{125}+\frac{23}{24}\cdot\frac{12}{125}=\frac1{10},
  \qquad
  \text{power}=\frac{64}{125}+\frac{23}{24}\cdot\frac{48}{125}=\frac{22}{25}.
\]
At $m=1$ the power is $2/5$. So a target of $p_*=4/5$ fails at two contributors
and is met at three, giving $q=3$. The powers at $m=4,5,6$ are $223/250$,
$4763/5000$ and $98311/100000$.

\subsection{Unequal categories and service-specific requirements}

Corollary~\ref{cor:occ} does not need equal occupancies. With a common
requirement $q$, nearly equal category sizes make $Kq\le M$ sufficient, since
their minimum size is $\lfloor M/K\rfloor$. With service-specific requirements
$q_c$, adequacy must instead be checked as $m_c\ge q_c$ for every category;
$\sum_c q_c\le M$ alone does not certify a particular assignment. A design below
the cap can still fail. In one worked case with $M=32$ and $q=3$, a nine-category
design has smallest pool $2$ and fails, although the cap is
$\lfloor 32/3\rfloor=10$.

\subsection{Family-wise allowances}

The convention throughout is per-category false-alarm control. To hold a total
family-wise allowance $\alpha$ instead, one conservative rule is
$a_j=\alpha/K_j$. For an existing category that allowance only tightens as the
menu grows, so the ordering proof of Lemma~\ref{lem:vorder} is unchanged, and no
independence across category alarms is needed for that union bound. The required
count then becomes $q(K)$ and the necessary condition is $K\,q(K)\le M$. The
fixed-$a$ value of $q$ from Section~\ref{sec:worked} does not carry over to a
family-wise convention.

\subsection{Necessary and sufficient welfare endpoints}

Where the exact loss $D_j$ is known, the endpoint $j_W$ of
Theorem~\ref{thm:budget} is exact, and this is the case under
Assumption~\ref{as:quad}, where $D=s\,e$. Where only the curvature sandwich of
Corollary~\ref{cor:mean} is available,
\[
  \frac L2 e_j\le\eps
  \quad\text{is sufficient},
  \qquad
  \frac\mu2 e_j\le\eps
  \quad\text{is necessary},
\]
and intersecting either with the verification condition gives respectively an
inner certificate and an outer restriction on the band. Neither should be called
the exact feasible set. For a menu whose profiles are not the conditional means,
use the actual squared allocation error of~\eqref{eq:eR} and not the
within-category variance.

\subsection{A menu-extension construction}

Lemma~\ref{lem:worder} and Lemma~\ref{lem:vorder} are stated for nested
refinements with mean profiles. A different construction inside the same
declare, provision and verify architecture satisfies
Theorem~\ref{thm:budget} without nesting or balance, and it is recorded here
because it makes clear how little the band theorem requires.

Let the menu grow, $\calR_0\subseteq\calR_1\subseteq\cdots\subseteq\calR_J$, with
every old profile retained on unchanged terms, and let each agent select a
utility-maximising offered profile under a fixed tie rule that preserves the
relative priority of the old profiles. Suppose every offered profile stays
occupied. Then three things hold at every level. Declaration gain is zero
pointwise, directly by Proposition~\ref{prop:gain}. Welfare loss is
non-increasing, because every old choice remains available so the pointwise
maximum cannot fall. And for each retained profile the category can only shrink,
because an agent that did not choose it before cannot begin choosing it when only
alternatives are added, its old choice being still available with the same
utility and tie priority. Applying Assumption~\ref{as:repro} to those shrinking
pools, and noting that the least-detectable category of the coarser level is
still monitored at the finer one,
gives $\pi_{j+1}\le\pi_j$ exactly as in Lemma~\ref{lem:vorder}.

Retaining old profiles does not keep them equal to the means of their current
members, and this construction is not claimed globally optimal at any
cardinality. It is a sufficient family for Theorem~\ref{thm:budget}, and it is
not a description of how a fitted menu is produced.

\section{The 96-agent construction}
\label{app:construction}
\subsection{\texorpdfstring{The optimum at every integer $K$}{The optimum at every integer K}}

Section~\ref{sec:construction} gives the closed form~\eqref{eq:Dstar} for
$K\mid M$. For general $K$, write
\[
  a_K=\lfloor M/K\rfloor,\qquad r_K=M-K a_K,
\]
so the balanced construction has $r_K$ groups of size $a_K+1$ and $K-r_K$ groups
of size $a_K$. By Lemma~\ref{lem:group} and Lemma~\ref{lem:balance} the exact
global minimum is
\begin{equation}
\label{eq:DstarK}
  D_K^*=\frac{(K-r_K)\bigl(a_K^3-a_K\bigr)
              +r_K\bigl[(a_K+1)^3-(a_K+1)\bigr]}{12M^3}.
\end{equation}
At $r_K=0$ this reduces to~\eqref{eq:Dstar}.

$D_K^*$ decreases strictly for $K<M$. For $K<M$, an optimal consecutive
construction has a group containing at least two distinct points. Splitting that
group into two nonempty consecutive groups and using their separate means
produces a candidate with loss $D_{\rm split}<D_K^*$. Hence
$D_{K+1}^*\le D_{\rm split}<D_K^*$. This compares optimum values across $K$. It
does not require the optimal partitions at successive $K$ to be nested, and in
general they are not. Their smallest occupancies $\lfloor M/K\rfloor$ are
non-increasing, which is what Corollary~\ref{cor:occ} uses.

\subsection{The self-selection boundary}

\begin{lemma}[Self-selection]
\label{lem:selfsel}
In the nearly balanced consecutive construction, every agent is strictly nearest
its own group's profile, so $g(t_i)=0$ for every $i$, at every integer
$1\le K\le M$.
\end{lemma}

\begin{proof}
Take adjacent groups of sizes $a$ and $b$ with $|a-b|\le1$. The midpoint of
their two centroids is displaced from the midpoint of the data gap between them
by $(b-a)/(4M)$, which is at most $1/(4M)$. The nearest demand point on either
side lies $1/(2M)$ from that gap midpoint. So the profile boundary stays strictly
between the two points, and each of them is closer to its own group's centroid.
The centroids increase in group order, so for a fixed $t$ the distance $|t-r_k|$
falls and then rises in $k$, and beating both neighbouring centroids is enough to
beat all of them. Applying the boundary calculation at every adjacent pair
therefore covers the ordered menu, and the case $K=1$ is trivial. The conclusion
$g\equiv0$ is Proposition~\ref{prop:gain} under Assumption~\ref{as:quad}.
\end{proof}

This is why the construction needs no rebalancing step. The intended memberships
are exactly the memberships self-selection produces, at every integer $K$ from
$1$ to $M$, including the cases where the group sizes are unequal.

\subsection{The entropy equality cases}

For the consecutive centroid construction, the grouping-cost bound is attained.
Equality in the subsequent arithmetic--geometric mean step holds exactly when the
occupied category probabilities are equal. If $K\mid M$ the balanced design has
$p_c=1/K$ throughout, $H(C)=\log_2 K$ and $D=D_K^*$, and
\[
  \frac1{12}\bigl(2^{-2H(C)}-M^{-2}\bigr)
  =\frac1{12}\Bigl(\frac1{K^2}-\frac1{M^2}\Bigr)=D_K^*,
\]
so~\eqref{eq:entlb} holds with equality. Setting the tolerance to $\eps=D_K^*$
therefore forces $H(C)\ge\log_2 K$ on every design meeting it, and the balanced
design attains that value. When $M/K\ge q$ it also meets verification and has
zero declaration gain, which is the sense in which $\log_2 K$ is a minimum over
alternative categorical designs and not the encoding cost of a preselected rule.

Away from the divisor cases the two sides separate. At $M=96$ and $\eps=0.002$
the floor evaluates to $2.68716$ bits, while the seven-category design that
attains the welfare optimum uses $2.80656$ bits, with five groups of $14$ and two
of $13$. The floor is a lower bound and not a claim that some partition attains
it.

Separately, on a deterministic nested construction the coarser category is a
function of the finer, so
\[
  H(C_{j+1})=H(C_j)+H(C_{j+1}\mid C_j)\ge H(C_j),
\]
with strict increase at any split of positive mass. That gives the band of
Theorem~\ref{thm:budget} an ordered entropy trace along the construction. Those
endpoint entropies are construction-relative and are not a bound over alternative
designs.

\subsection{What the declaration message must distinguish}

Fix a deterministic delivered allocation $A=\alpha(T)$ with the menu known to
both parties, and suppose the coordinator must reproduce $A$ exactly from a
per-agent message $Z$ with no other demand-related channel. Then $A$ is a
function of $Z$, so
\[
  I(T;Z)=I(A;Z)+I(T;Z\mid A)\ge H(A),
\]
and sending the allocation identifier attains equality. A fixed-length message
needs at least $\lceil\log_2 L\rceil$ bits for $L$ distinct delivered profiles of
positive probability. If every occupied category has a distinct profile then
$H(A)=H(C)$ and $L=K$. If categories share a profile the allocation identifier
can be coarser, and whether the finer category must be retained for a separate
verification obligation is a different question.

This accounting covers the declaration and allocation channel only. Menu
construction, profile publication, category identifiers attached to telemetry,
aggregate outcome counts and repeated lifecycle messages are outside it.

\subsection{A general finite-population obstruction}

Corollary~\ref{cor:floor} is stated for the uniform grid. A distribution-free
version holds under an explicit separation condition. Let $x_1,\dots,x_M$ be
distinct oracle allocations in $\RR^d$ separated by
$\|x_i-x_j\|\ge\theta>0$ for $i\neq j$, with
equal welfare weights and the curvature lower bound of
Assumption~\ref{as:curv}. For a category with $m$ members and any profile $r$,
\[
  \sum_{i}\|x_i-r\|^2\;\ge\;\frac1m\sum_{i<j}\|x_i-x_j\|^2\;\ge\;\frac{(m-1)\theta^2}{2}.
\]
Summing over $K$ categories and dividing by $M$ gives
$D\ge(\mu\theta^2/4)(1-K/M)$, and if every category needs $q$ contributors,
\begin{equation}
\label{eq:sepfloor}
  D\;\ge\;\frac{\mu\theta^2}{4}\Bigl(1-\frac{\lfloor M/q\rfloor}{M}\Bigr)
  \;\ge\;\frac{\mu\theta^2}{4}\Bigl(1-\frac1q\Bigr),
\end{equation}
which is strictly positive for $q>1$, in any dimension, for arbitrary partitions
and profiles. The premises are substantive. If some agents share an oracle
allocation then $\theta=0$ and~\eqref{eq:sepfloor} says nothing, which is correct,
since identical demands can be grouped and monitored at no allocation cost.

\section{A general-source information floor}
\label{app:floor}
Corollary~\ref{cor:info} needs a source-dependent lower bound on category
entropy. Section~\ref{sec:construction} supplies one for the uniform scalar grid.
This appendix states the general form used there.

\subsection{Finite weighted sources}

Let $X=x(T)$ take finitely many distinct values in $\RR^d$, with $x_i$ carrying
probability $w_i>0$. Define
\[
  \gamma_P=\min_{i<j}\frac{\|x_i-x_j\|}{w_i^{1/d}+w_j^{1/d}},
  \qquad
  A_P=\frac{d}{d+2}\,\gamma_P^2,
  \qquad
  S_P=\sum_i w_i^{1+2/d}.
\]
Then every category rule satisfies
\begin{equation}
\label{eq:floorfinite}
  e\;\ge\;A_P\Bigl(2^{-2H(C)/d}-S_P\Bigr),
\end{equation}
using $e\ge0$ when the right side is negative. Under
Assumption~\ref{as:curv} this converts into a bound on welfare loss through
$D\ge(\mu/2)e$, and inverting gives the floor used in~\eqref{eq:INFO},
\begin{equation}
\label{eq:RP}
  R_P(\eps)=\max\left\{0,\;
  \frac d2\log_2\frac1{S_P+2\eps/(\mu A_P)}\right\}.
\end{equation}

For completeness, we give the packing calculation. Let $v_d$ be the volume of
the unit ball in $\RR^d$. If a measure of mass $p$ has density at most $F$, then
at most $F v_d u^d$ of its mass can lie within distance $u$ of any profile $r$.
With $R=(p/(F v_d))^{1/d}$, this gives
\[
  \int\|y-r\|^2\,d\nu(y)
  \ge\int_0^R 2u\bigl(p-Fv_du^d\bigr)\,du
  =\frac{d}{d+2}(Fv_d)^{-2/d}p^{1+2/d}.
\]

For the proof only, replace each source point $x_i$ by a uniform ball of radius
$\gamma_P w_i^{1/d}$ and retain its probability $w_i$. The balls have disjoint
interiors by the definition of $\gamma_P$. The resulting variable $V$ has
density at most $F=(v_d\gamma_P^d)^{-1}$. Generate the centred displacement
independently of the category conditional on the source point, and retain the
original category and profile. Its mean is zero and its total added squared
error is $A_P S_P$. Therefore
\[
  \EE\|V-r_C\|^2=e+A_PS_P.
\]

For category $c$, the joint measure of $V$ and the event $C=c$ has mass $p_c$
and density at most $F$. Applying the preceding moment bound to each category
and summing gives
\[
  e+A_PS_P\ge A_P\sum_c p_c^{1+2/d}.
\]

Weighted arithmetic--geometric mean yields
\[
  \sum_c p_c^{1+2/d}
  =\sum_c p_c\,p_c^{2/d}
  \ge\prod_c(p_c^{2/d})^{p_c}
  =2^{-2H(C)/d}.
\]

Subtracting $A_P S_P$ proves the stated allocation-error bound. The Welfare
Bound gives $D\ge(\mu/2)e$. Rearranging under $D\le\eps$, together with
$H(C)\ge0$, gives the displayed information floor. The artificial spread is a
proof device and does not change the declaration or allocation mechanism. A
source with only one preferred allocation requires no allocation information and
is treated separately.

\subsection{Bounded-density continuous sources}

For a continuous oracle source with density bounded above by $F$, write $v_d$
for the volume of the unit ball and
\[
  A_F=\frac d{d+2}\,(F v_d)^{-2/d}.
\]
For each category, its joint source measure has mass $p_c$ and density at most
$F$. The moment bound just proved therefore gives
$e\ge A_F\sum_c p_c^{1+2/d}\ge A_F\,2^{-2H(C)/d}$. Applying the Welfare Bound
and rearranging gives the corresponding entropy floor for positive tolerances.
This statement uses a bounded density on the stated Euclidean space and does not
automatically cover a singular source.

\subsection{How the floor is used}

The floor is a converse. It says that any design meeting a welfare tolerance must
carry at least a certain category entropy, whatever partition it uses and
whatever profiles it delivers. Paired with the occupancy ceiling of
Corollary~\ref{cor:occ} it gives the envelope~\eqref{eq:INFO}, and non-overlap of
the two sides certifies that no design meets both targets. Overlap certifies
nothing on its own. An entropy value between the two bounds says nothing about
whether some partition attains it with adequate minimum occupancy.
Section~\ref{sec:construction} is where an attained statement is produced, by
exhibiting the design rather than by inspecting the envelope.

\begin{remark}
Section~\ref{sec:construction} does not rely on~\eqref{eq:floorfinite}. Its floor
is proved directly for the uniform scalar grid, by Lemma~\ref{lem:group} and the
weighted arithmetic--geometric mean step of Corollary~\ref{cor:entropy}, and it
is attained there. The statements in this appendix are the general-source
versions, at the generality of the packing argument, and they do not aim at the
best constant.
\end{remark}

\section{Reproducibility of the numerical illustration}
\label{app:repro}
\subsection{Demand generator}

The synthetic population of Section~\ref{sec:numerical} is generated as follows.
Each run uses a fresh NumPy \texttt{default\_rng(seed)} with seeds 1--10, drawing
50{,}000 sessions from a mixture of five service types with weights:
video streaming 0.35, interactive communications 0.15, IoT telemetry 0.20, web
browsing 0.20, bulk transfer 0.10.

The four-dimensional demand vector for each session is
\[
  \Bigl(\min(\mathrm{tp}/100,\,1),\;\;
  1-\min(\mathrm{lat}/500,\,1),\;\;
  \mathrm{prb},\;\;
  \mathrm{mcs}/3\Bigr),
\]
where throughput (Mbps) is lognormal, latency (ms) is normal clamped at zero,
physical resource block (PRB) utilisation is beta-distributed, and modulation
and coding scheme (MCS) order is a categorical draw over $\{0,1,2,3\}$,
indexing QPSK, 16QAM, 64QAM and 256QAM, all with type-specific parameters.
Latency is inverted so that lower latency corresponds to higher demand. The
normalisation uses fixed known bounds (100\,Mbps, 500\,ms, 3 for MCS), not
data-driven minima.

Each session draws its service type from the mixture weights above, then draws
its four coordinates independently from that type's own distributions. The
generator imposes no further dependence between coordinates, so within a type
the four are independent, and the type mixture is the only source of
correlation between them in the pooled population.

The five type-specific distribution parameters are:

\begin{center}
\begin{tabular}{lccc}
\toprule
Type & Throughput (lognormal) & Latency (normal) & PRB (beta) \\
\midrule
Video streaming   & mean 25, $\sigma$ 0.5  & 50, std 10   & $a=5,\,b=2$ \\
Interactive comms & mean 5,  $\sigma$ 0.6  & 10, std 3    & $a=2,\,b=4$ \\
IoT telemetry     & mean 0.1, $\sigma$ 1.0 & 200, std 50  & $a=1,\,b=8$ \\
Web browsing      & mean 10, $\sigma$ 0.8  & 30, std 8    & $a=2,\,b=3$ \\
Bulk transfer     & mean 50, $\sigma$ 0.4  & 500, std 100 & $a=7,\,b=1.5$ \\
\bottomrule
\end{tabular}
\end{center}

For lognormal demands, the internal $\mu$ is set to
$\ln(\text{mean})-0.5\sigma^2$ so that the stated mean is the target mean of the
distribution.

The MCS profiles are type-specific categorical distributions on the four
modulation orders, in the order QPSK, 16QAM, 64QAM, 256QAM:

\begin{center}
\begin{tabular}{lcccc}
\toprule
 & \multicolumn{4}{c}{MCS probability} \\
\cmidrule(lr){2-5}
Type & QPSK & 16QAM & 64QAM & 256QAM \\
\midrule
Video streaming   & 0.05 & 0.10 & 0.30 & 0.55 \\
Interactive comms & 0.10 & 0.25 & 0.40 & 0.25 \\
IoT telemetry     & 0.60 & 0.25 & 0.10 & 0.05 \\
Web browsing      & 0.10 & 0.20 & 0.40 & 0.30 \\
Bulk transfer     & 0.05 & 0.10 & 0.25 & 0.60 \\
\bottomrule
\end{tabular}
\end{center}

The entries of both tables are the coordinate choices of the synthetic model.
They are not measurements of a deployed network, and nothing here claims that
the mechanism has been validated against one.

\subsection{Category rules}

At $K=3$, the five types are grouped by resource intensity into three semantic
categories: \emph{high throughput} (video streaming, bulk transfer),
\emph{low latency} (interactive communications, web browsing), and
\emph{low throughput} (IoT telemetry). Each category receives the componentwise
mean of its members' demand vectors.

At $K=10$ and $K=30$, categories are fitted by $k$-means
(\texttt{scikit-learn}, 10 restarts, 300 maximum iterations,
$k$-means++ initialisation). The random seed for each fit is drawn
from the run's generator. Each participant receives the fitted cluster
centroid for its assigned label.

\subsection{Reported statistics}

For each run $b$ and granularity, the reported quantities are:
\begin{itemize}
\item $D_b$: mean squared error against the delivered profiles,
  $\frac1N\sum_i\|t_{bi}-r_{c_i}\|^2$.
\item $V_b$: total demand variance,
  $\frac1N\sum_i\|t_{bi}-\bar t_b\|^2$.
\item $D_b/V_b$: the normalised ratio, computed per run.
\item Mean~$g$: mean declaration gain across all $N$ participants in run $b$.
\item Max~$g$: maximum declaration gain across all $N$ participants in run $b$.
\end{itemize}

Across the ten runs, $D$ and mean~$g$ are reported as the arithmetic mean of the
per-run values. The $D/\mathrm{Var}(T)$ column reports the mean and standard
deviation of the per-run ratios. Max~$g$ is the overall maximum across all
participants in all ten runs.

Fitted profiles need not equal the conditional means of their final memberships.
Writing $x=x(T)$, the squared allocation error against the delivered profiles
decomposes as
\begin{equation}
\label{eq:eR}
  \EE\|x-r_{C}\|^2=\EE\|x-\EE[x\mid C]\|^2+\EE\|\EE[x\mid C]-r_C\|^2,
\end{equation}
where the second term vanishes for centroid profiles and not otherwise. The
profile-to-mean offset reported in Section~\ref{sec:numerical} estimates that
second term as the weighted mean of $\|r_c-\bar x_c\|^2$ across categories within
a run, then averaged across runs. All reported allocation errors use the profiles
actually delivered.

Declaration gain is the delivered squared distance minus the minimum squared
distance over the same offered profiles, with ties accepted via a tolerance of
$10^{-12}+10^{-10}\max(\text{current},\text{best})$.

\subsection{Software versions}

The results were produced with Python 3.12, NumPy 2.4.3, and scikit-learn 1.8.0.
The generator and fitting code, including the exact parameter values above, is
available from the author upon request.

\section{A fixed-menu consequence}
\label{app:fixedmenu}
This appendix records two consequences of Proposition~\ref{prop:gain} for a
fixed menu. Neither is needed for Theorems~\ref{thm:welfare},
\ref{thm:aggdom}, \ref{thm:budget} or~\ref{thm:exactband}, and neither is
claimed as a contribution of this paper. They are collected here because they
are the standard reading of the proposition.

\subsection{A fixed-menu optimisation consequence}

Let $D_K^{\mathrm{free}}$ be the infimum of allocation loss over menus with at
most $K$ profiles and measurable assignments. Let $D_K^{\mathrm{IC}}$ be the
infimum over the same class with zero declaration gain imposed.

\begin{corollary}[No welfare price at the optimum]
\label{cor:free}
If the admissible class is closed under reassignment to a utility-maximising
offered profile, then
\begin{equation}
\label{eq:ICfree}
  D_K^{\mathrm{IC}}=D_K^{\mathrm{free}}.
\end{equation}
\end{corollary}

\begin{proof}
Restricting the feasible set gives $D_K^{\mathrm{free}}\le D_K^{\mathrm{IC}}$.
Conversely, keep the menu of any unrestricted candidate and assign each type to
its preferred offered profile. Allocation loss does not increase, the number of
offered profiles does not change, and Proposition~\ref{prop:gain} gives zero
gain. Taking infima proves the reverse inequality.
\end{proof}

The equality concerns the stated fixed-menu class. Reassignment can change
category probabilities and occupancies, so it need not preserve an entropy
restriction, a balance requirement, or a monitoring guarantee. The exact
construction of Section~\ref{sec:construction} establishes its allocation and
verification properties on the same memberships.

\subsection{The exact gain identity}

Proposition~\ref{prop:gain} bounds the declaration gain by the allocation loss.
The exact relation is finer, because it separates two defects that the bound
combines.

For a fixed menu $\calM$, write $c_{\calM}(t)\in\arg\min_k \ell(t,r_k)$ for the
best-response assignment on that same menu. Then, pointwise,
\begin{equation}
\label{eq:gainid}
  g(t)=\ell(t,r_{c(t)})-\min_k \ell(t,r_k),
  \qquad
  \EE\,g(T)=D(\calM,c)-D(\calM,c_{\calM}),
\end{equation}
and hence, with $D_K^{\mathrm{free}}$ as in Corollary~\ref{cor:free},
\[
  0\;\le\;\EE\,g(T)\;\le\;D(\calM,c)-D_K^{\mathrm{free}} .
\]
The first equality holds because $\ell(t,r)=U(t,x(t))-U(t,r)$ differs from
$-U(t,r)$ by a term not depending on $r$, so maximising $U(t,\cdot)$ over the
menu is minimising $\ell(t,\cdot)$ over it. The second follows by taking
expectations, and the third from $D(\calM,c_{\calM})\ge D_K^{\mathrm{free}}$.

The reading is that declaration gain is exactly the loss recoverable by fixing
the assignment while leaving the offered profiles alone. Call
$\min_k\ell(t,r_k)$ the coverage loss, the shortfall that no declaration can
avoid because no offered profile fits well, and call the remainder the
assignment loss. Only the assignment loss is a profitable deviation. The bound
$\EE g\le D$ of Proposition~\ref{prop:gain} charges the deviation with both,
and can therefore be loose.

\bibliographystyle{plainnat}
\bibliography{refs}

@article{shannon1948,
  author  = {Shannon, Claude E.},
  title   = {A Mathematical Theory of Communication},
  journal = {Bell System Technical Journal},
  volume  = {27},
  number  = {3},
  pages   = {379--423},
  year    = {1948},
  doi     = {10.1002/j.1538-7305.1948.tb01338.x}
}

@book{berger1971,
  author    = {Berger, Toby},
  title     = {Rate Distortion Theory: A Mathematical Basis for Data
               Compression},
  publisher = {Prentice-Hall},
  year      = {1971}
}

@article{gray_neuhoff1998,
  author  = {Gray, Robert M. and Neuhoff, David L.},
  title   = {Quantization},
  journal = {IEEE Transactions on Information Theory},
  volume  = {44},
  number  = {6},
  pages   = {2325--2383},
  year    = {1998},
  doi     = {10.1109/18.720541}
}

@article{lloyd1982,
  author  = {Lloyd, Stuart P.},
  title   = {Least Squares Quantization in {PCM}},
  journal = {IEEE Transactions on Information Theory},
  volume  = {28},
  number  = {2},
  pages   = {129--137},
  year    = {1982},
  doi     = {10.1109/TIT.1982.1056489}
}

@article{neyman_pearson1933,
  author  = {Neyman, Jerzy and Pearson, Egon S.},
  title   = {On the Problem of the Most Efficient Tests of Statistical
             Hypotheses},
  journal = {Philosophical Transactions of the Royal Society of London,
             Series A},
  volume  = {231},
  pages   = {289--337},
  year    = {1933},
  doi     = {10.1098/rsta.1933.0009}
}

@book{lehmann_romano2005,
  author    = {Lehmann, E. L. and Romano, Joseph P.},
  title     = {Testing Statistical Hypotheses},
  edition   = {3rd},
  publisher = {Springer},
  year      = {2005},
  doi       = {10.1007/0-387-27605-X}
}

@article{blackwell1953,
  author  = {Blackwell, David},
  title   = {Equivalent Comparisons of Experiments},
  journal = {The Annals of Mathematical Statistics},
  volume  = {24},
  number  = {2},
  pages   = {265--272},
  year    = {1953},
  doi     = {10.1214/aoms/1177729032}
}

@inproceedings{khodabakhsh2024,
  author    = {Khodabakhsh, Ali and Pountourakis, Emmanouil and
               Taggart, Samuel},
  title     = {Simple Delegated Choice},
  booktitle = {Proceedings of the 2024 Annual ACM-SIAM Symposium on Discrete
               Algorithms (SODA)},
  pages     = {569--590},
  publisher = {SIAM},
  year      = {2024},
  doi       = {10.1137/1.9781611977912.21}
}

@article{bns2007,
  author  = {Blumrosen, Liad and Nisan, Noam and Segal, Ilya},
  title   = {Auctions with Severely Bounded Communication},
  journal = {Journal of Artificial Intelligence Research},
  volume  = {28},
  pages   = {233--266},
  year    = {2007},
  doi     = {10.1613/jair.2081}
}

@article{shlezinger2019,
  author  = {Shlezinger, Nir and Eldar, Yonina C. and Rodrigues,
             Miguel R. D.},
  title   = {Hardware-Limited Task-Based Quantization},
  journal = {IEEE Transactions on Signal Processing},
  volume  = {67},
  number  = {20},
  pages   = {5223--5238},
  year    = {2019},
  doi     = {10.1109/TSP.2019.2935864}
}

@inproceedings{tishby1999,
  author        = {Tishby, Naftali and Pereira, Fernando C. and Bialek,
                   William},
  title         = {The Information Bottleneck Method},
  booktitle     = {Proc. 37th Annual Allerton Conference on Communication,
                   Control, and Computing},
  pages         = {368--377},
  year          = {1999},
  eprint        = {physics/0004057},
  archivePrefix = {arXiv}
}

\end{document}